\documentclass[11pt,onecolumn]{article}
\usepackage[utf8]{inputenc} 
\usepackage[T1]{fontenc}    
\usepackage{hyperref}       
\usepackage{url}            
\usepackage{booktabs}       
\usepackage{amsfonts}       
\usepackage{nicefrac}       
\usepackage{microtype}      
\usepackage{times}
\usepackage{amsmath}
\usepackage{mathrsfs}
\usepackage{amssymb}
\usepackage{multirow}
\usepackage{xspace}
\usepackage{theorem}
\usepackage{graphicx}
\usepackage{ifpdf}
\usepackage{latexsym}
\usepackage{euscript}
\usepackage{xspace}
\usepackage{color}
\usepackage{makeidx}
\usepackage{wrapfig}

\long\def\remove#1{}

\newtheorem{theorem}{Theorem}[section] 
\newtheorem{lemma}[theorem]{Lemma}

\newtheorem{cor}[theorem]{Corollary}

\newtheorem{definition}[theorem]{Definition}
\newenvironment{proof}{{\em Proof:}}{\hfill{\hfill\rule{2mm}{2mm}}}

\definecolor{darkred}{rgb}{0.8, 0.2, 0.2}
\definecolor{darkgreen}{rgb}{0.5, 0.8, 0.1}
\definecolor{darkpurple}{rgb}{1.0, 0, 1.0}
\definecolor{darkblue}{rgb}{0, 0, 1.0}

\newcommand {\mm}[1] {\ifmmode{#1}\else{\mbox{\(#1\)}}\fi}

\newcommand{\reals}     {{\mathbb{R}}}
\newcommand{\weightfunc}    {{weight-function}\xspace}

\newcommand{\PIWK}          {{WKPI}}
\newcommand{\dgm}           {\mathrm{Dg}}
\newcommand{\pers}          {\mathrm{pers}}
\newcommand{\awfunc}        {{\alpha}}
\newcommand{\aD}            {A} 
\newcommand{\aP}            {{\mathcal{P}}} 
\newcommand{\ap}            {{\mathbf{p}}} 
\newcommand{\aI}            {{\mathrm{PI}}}
\newcommand{\aC}            {{\mathcal{C}}}

\newcommand{\myitnum}           {{\mathrm{R}}}
\newcommand{\ourw}          {{\omega}}
\newcommand{\kw}            {{k_w}}
\newcommand{\ourD}          {{\mathrm{D_\omega}}}

\newcommand{\acost}         {{cost}}
\newcommand{\TC}            {{TC}}
\newcommand{\trace}         {{\mathrm{Tr}}}
\newcommand{\altPIWK}       {{altWKPI}}

\newcommand{\trainPWGK}          {{trainPWGK}}

\newcommand{\fullversiononly}[1]        {{#1}}

\title{Learning metrics for persistence-based summaries and applications for graph classification}
\author{Qi Zhao$^*$\\ zhao.2017@osu.edu \and  Yusu Wang \thanks{Computer Science and Engineering Department, The Ohio State University, Columbus, OH 43221, USA. } \\ yusu@cse.ohio-state.edu }
\date{}

\begin{document}

\maketitle

\begin{abstract}
Recently a new feature representation framework
based on a topological tool called persistent homology (and its persistence diagram summary) has gained much momentum. 
A series of methods have been developed to map a persistence diagram to a vector representation so as to facilitate the downstream use of machine learning tools. In these approaches, the importance (weight) of different persistence features are usually \emph{pre-set}. 
However often in practice, the choice of the \weightfunc{} should depend on the nature of the specific data at hand. It is thus highly desirable to \emph{learn} a best \weightfunc{} (and thus metric for persistence diagrams) from labelled data. We study this problem and develop a new weighted kernel, called \emph{\PIWK}, for persistence summaries, 
as well as an optimization framework to learn the weight (and thus kernel). 
We apply the learned kernel to the challenging task of graph classification, and show that our \PIWK-based classification framework obtains similar or (sometimes significantly) better results than {\bf the best results} from a range of previous graph classification frameworks on benchmark datasets.
\end{abstract}

\section{Introduction}

In recent years a new data analysis methodology based on a topological tool called persistent homology has started to attract momentum. 
The persistent homology is one of the most important developments in the field of topological data analysis, 
and there have been fundamental developments both on the theoretical front (e.g, \cite{ELZ02,CZ09,CCG09,CS10,CSGO16, bhatia2018understanding}), and on algorithms / implementations (e.g, \cite{Sheehy12,PHAT,GUDHI,Simba,Kerber2018,Ripser}). 
On the high level, given a domain $X$ with a function $f: X \to \reals$ on it, the persistent homology summarizes ``features'' of $X$ across multiple scales simultaneously in a single summary called the \emph{persistence diagram} (see the second picture in Figure \ref{fig:perpipeline}). 
A persistence diagram consists of a multiset of points in the plane, where each point $p = (b, d)$ intuitively corresponds to the birth-time ($b$) and death-time ($d$) of some (topological) features of $X$ w.r.t. $f$. 
Hence it provides a concise representation of $X$, capturing \emph{multi-scale features} of it simultaneously. 
Furthermore, the persistent homology framework can be applied to complex data (e.g, 3D shapes, or graphs), and different summaries could be constructed by putting different descriptor functions on input data. 

Due to these reasons, a new persistence-based feature vectorization and data analysis framework (Figure \ref{fig:perpipeline}) has become popular. Specifically, given a collection of objects, say a set of graphs modeling chemical compounds, one can first convert each shape to a persistence-based representation. The input data can now be viewed as a set of points in a persistence-based feature space. Equipping this space with appropriate distance or kernel, one can then perform downstream data analysis tasks (e.g, clustering). 

The original distances for persistence diagram summaries unfortunately do not lend themselves easily to machine learning tasks. Hence in the last few years, starting from the persistence landscape \cite{bubenik2015Statistical}, there have been a series of methods developed to map a persistence diagram to a vector representation to facilitate machine learning tools \cite{Reininghaus2015A,Adams2017Persistence,Kusano2017Kernel,Carri2017Sliced,Levie2017CayleyNets}. 
\fullversiononly{Recent ones include Persistence Scale-Space kernel \cite{Reininghaus2015A}, Persistence Images \cite{Adams2017Persistence}, Persistence Weighted Gaussian kernel (PWGK) \cite{Kusano2017Kernel}, Sliced Wasserstein kernel \cite{Carri2017Sliced}, and Persistence Fisher kernel \cite{Le2018Persistence}. }

\begin{figure}[tbp]
\centering \includegraphics[height = 3.5cm, width = 14cm]{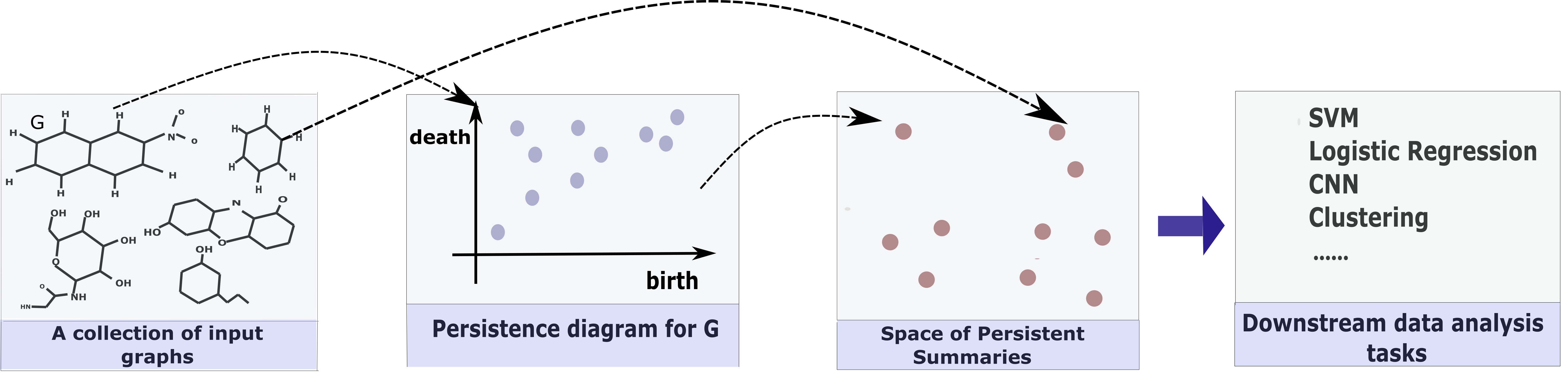}
\vspace*{-0.2in}
\caption{A persistence-based data analysis framework. 
\label{fig:perpipeline}}
\end{figure}

In these approaches, when computing the distance or kernel between persistence summaries, the importance (weight) of different persistence features are often \emph{pre-determined}. 
In persistence images \cite{Adams2017Persistence} and PWGK \cite{Kusano2017Kernel}, the importance of having a {\bf \weightfunc{}} for the birth-death plane (containing the persistence points) has been emphasized and explicitly included in the formulation of their kernels. 
However, before using these kernels, the \weightfunc{} needs to be pre-set. 

On the other hand, as recognized by \cite{Hofer2017Deep}, the choice of the \weightfunc{} should depend on the nature of the specific type of data at hand. For example, for the persistence diagrams computed from atomic configurations of molecules, features with small persistence could capture the local packing patterns which are of utmost importance and thus should be given a larger weight; while in many other scenarios, small persistence leads to noise with low importance. 
However, in general researchers performing data analysis tasks may not have such prior insights on input data. Thus it is natural and highly desirable to \emph{learn} a best \weightfunc{} from labelled data. 

\paragraph{Our work.} 
We study the problem of learning an appropriate metric (kernel) for persistence summaries from labelled data, and apply the learnt kernel to the challenging graph classification task. 

{\it (1) Metric learning for persistence summaries:~} 
We propose a new weighted-kernel (called \emph{\PIWK}), for persistence summaries based on persistence images representations. Our \PIWK{} kernel is positive semi-definite and its induced distance is stable.
A \weightfunc{} used in this kernel directly encodes the importance of different locations in the persistence diagram. 
We next model the metric learning problem for persistence summaries as the problem of learning (the parameters of) this \emph{\weightfunc{}} from a certain function class. 
In particular, the metric-learning is formulated as an optimization problem over a specific cost function we propose. This cost function has a simple matrix view which helps both conceptually clarify its meaning and simplify the implementation of its optimization.  

{\it (2) Graph classification application:~}
Given a set of objects with class labels, we first learn a best \PIWK{}-kernel as described above, and then use the learned \PIWK{} to further classify objects. We implemented this \emph{\PIWK{}-classification framework}, and apply it to a range of graph data sets. 
Graph classification is an important problem, and 
there has been a large literature on developing effective graph representations (e.g, \cite{Hido2009A, Shervashidze2009A, Bai2015A, Kriege2016On, shervashidze2011weisfeiler, Xu2015A,Neumann2012Efficient}, including the very recent persistent-homology enhanced WL-kernel \cite{riech2019persistent}), and graph neural networks (e.g, graph neural networks \cite{Yanardag2015Deep, Niepert2016Learning, Xu2018b, Verma2017, Levie2017CayleyNets,KondorSPAT18}) to classify graphs. 
The problem is challenging as graph data are less structured. We perform our \PIWK{}-classification framework on various benchmark graph data sets as well as new neuron-cell data sets. 
Our learnt \PIWK{} performs consistently better than other persistence-based kernels. 
Most importantly, when compared with existing state-of-the-art graph classification frameworks, our framework shows similar or (sometimes significantly) better performance in almost all cases than the \emph{best results} by existing approaches. 

\noindent We note that \cite{Hofer2017Deep} is the first to recognize the importance of using labelled data to learn a task-optimal representation of topological signatures. They developed an end-to-end deep neural network for this purpose, using a novel and elegant design of the input layer to implicitly learn a task-specific representation. Very recently, in a parallel and independent development of our work, Carri\`{e}re et al. \cite{carriere2019general} built an interesting new neural network based on the DeepSet architecture \cite{zaheer2017deep}, which can achieve an end-to-end learning for multiple persistence representations \emph{in a unified manner}. 
Compared to these developments, we instead explicitly formulate the metric-learning problem for persistence-summaries, and decouple the metric-learning (which can also be viewed as representation-learning) component from the downstream data analysis tasks. 
Also as shown in Section \ref{sec:exp}, our \PIWK{}-classification framework (using SVM) achieves better results on graph classification datasets. 

\begin{figure*}[t]
\begin{center}
\begin{tabular}{ccc}
\includegraphics[height=2.6cm, width=6.5cm]{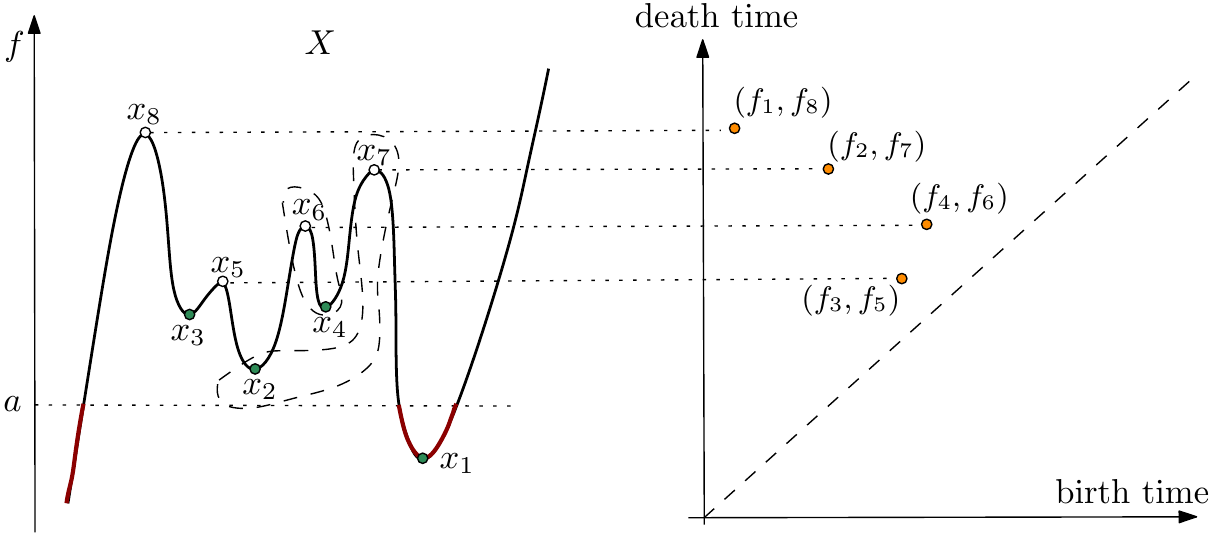} &\includegraphics[height=2.6cm, width=3.5cm]{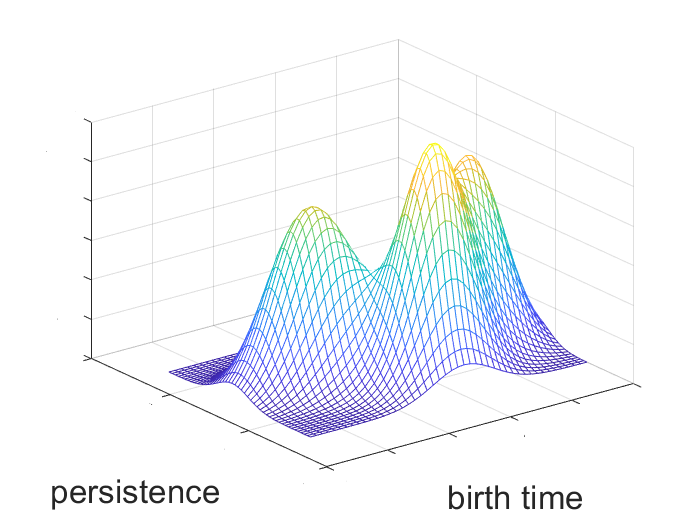} & \includegraphics[height=2.6cm, width = 3.0cm]{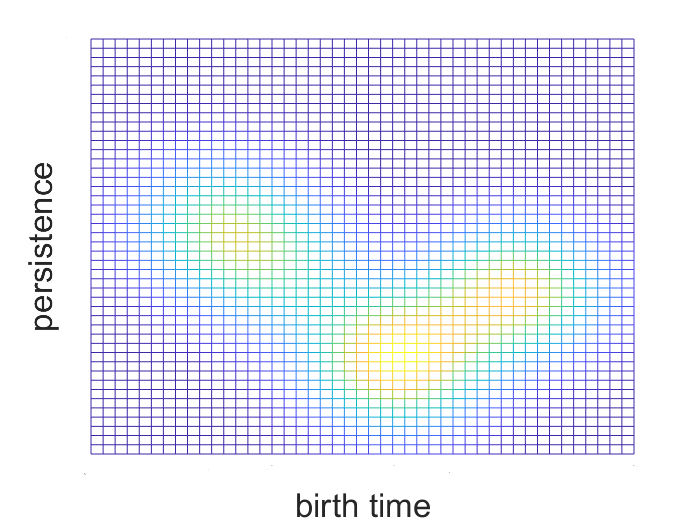}\\
(a) &  (b)  & (c)
\end{tabular}
\end{center}
\vspace*{-0.1in}\caption{(a): As we sweep the curve bottom-up in increasing $f$-values, at certain critical moments new 0-th homological features (connected components) are created, or destroyed (i.e, components merge). For example, a component is created when passing $x_4$ and killed when passing $x_6$, giving rise to the persistence-point $(f_4, f_6)$ in the persistence diagram ($f_i := f(x_i)$).   
(b) shows the graph of a persistence surface (where $z$-axis is the function $\rho_A$), and (c) is its corresponding persistence image. 
\label{fig:funcper}}
\end{figure*}

\section{Persistence-based framework}
\label{sec:background}

We first give an informal description of persistent homology below. See \cite{EH10} for more detailed exposition on the subject. 

Suppose we are given a shape $X$ (in our later graph classification application, $X$ is a graph). 
Imagine we inspect $X$ through a \emph{filtration of $X$}, which is a sequence of growing subsets of $X$: $X_1 \subseteq X_2 \subseteq \cdots \subseteq X_n = X$. 
As we scan $X$, sometimes a new feature appears in $X_i$, and sometimes an existing feature  disappears upon entering $X_j$. 
Using the topological object called homology classes to describe these features (intuitively components, independent loops, voids, and their high dimensional counter-parts), the birth and death of topological features can be captured by the \emph{persistent homology}, in the form of a \emph{persistence diagram} $\dgm X$. 
Specifically, for each dimension $k$, $\dgm_k X$ consists of a multi-set of points in the plane (which we call the \emph{birth-death plane $\reals^2$}): each point ($b,d$) in it, called a \emph{persistence-point}, indicates that a certain $k$-dimensional homological feature is created upon entering $X_b$ and destroyed upon entering $X_d$. 
In the remainder of the paper, we often omit the dimension $k$ for simplicity: when multiple dimensions are used for persistence features, we will apply our construction to each dimension and concatenate the resulting vector representations. 

A common way to obtain a meaningful filtration of $X$ is via the \emph{sublevel-set filtration} induced by a \emph{descriptor function} $f$ on $X$. 
More specifically, given a function $f: X \to \reals$, let $X_{\le a} := \{ x\in X \mid f(x) \le a\}$ be its \emph{sublevel-set at $a$}. 
Let $a_1 < a_2 < \cdots < a_n$ be $n$ real values. The sublevel-set filtration w.r.t. $f$ is:  
    $X_{\le a_1} \subseteq X_{\le a_2} \subseteq \cdots \subseteq X_{\le a_n}; $
and its persistence diagram is denoted by $\dgm f$. 
Each persistence-point $p=(a_i, a_j) \in \dgm f$ indicates the function values when some topological features are created (when entering $X_{\le a_i}$) and destroyed (in $X_{\le a_j}$), and the \emph{persistence} of this feature is its life-time $\pers(p) = |a_j - a_i|$. 
See Figure \ref{fig:funcper} (a) for a simple example where $X = \reals$. 
If one sweeps $X$ top-down in decreasing function values, one gets the persistence diagram induced by the super-levelset filtration of $X$ w.r.t. $f$ in an analogous way. Finally, if one tracks the change of topological features in the \emph{levelset $f^{-1}(a)$}, one obtains the so-called \emph{levelset zigzag persistence} \cite{CSM09} (which contains the information captured by the \emph{extended persistence} \cite{CEH09}). 

\paragraph{Graph Setting.} 
 Given a graph $G=(V, E)$, a descriptor function $f$ defined on $V$ or $E$ will induce a filtration and its persistence diagrams. Suppose $f: V \to \reals$ is defined on the node set of $G$ (e.g, the node degree). Then we can extend $f$ to $E$ by setting $f(u,v) = max\{f(u), f(v)\}$, and the sublevel-set at $a$ is defined as $G_{\le a}:= \{ \sigma \in V \cup E \mid f(\sigma) \le a\}$. Similarly, if we are given $f: E \to \reals$, then we can extend $f$ to $V$ by setting $f(u) = \min_{u\in e, e\in E} f(e)$. 
When scanning $G$ via the sublevel-set filtration of $f$, connected components in the swept subgraphs will be created and merged, and new cycles will be created. 
The formal events are encoded in the 0-dimensional persistence diagram. The the 1-dimensional features (cycles), however, we note that cycles created will never be killed, as they are present in the total space $X = G$. To this end, we use the so-called \emph{extended persistence} introduced in \cite{CEH09} which can record information of cycles. 

Now given a collection of shapes $\Xi$, we can compute a persistence diagram $\dgm X$ for each $X \in \Xi$, which maps the set $\Xi$ to a set of points in the space of persistence diagrams. 
There are natural distances defined for persistence diagrams, including the bottleneck distance and the Wasserstein distance, both of which have been well studied (e.g, stability under them \cite{CEH07,CEHY10,CSGO16}) with efficient implementations available \cite{KD17, Hera}. 
However, to facilitate downstream machine learning tasks, it is desirable to further map the persistence diagrams to another ``vector'' representation. 
Below we introduce one such representation, called the persistence images \cite{Adams2017Persistence}, as our new kernel is based on it. 

\fullversiononly{\paragraph{Persistence images.}} 
Let $\aD$ be a persistence diagram (containing a multiset of persistence-points). Following \cite{Adams2017Persistence}, set $T: \reals^2 \to \reals^2$ to be the linear transformation\footnote{In fact, we can define our kernel without transforming the persistence diagram. We use the transformation simply to follow the same convention as persistence images.} where for each $(x,y) \in \reals^2$, $T(x, y) = (x, y-x)$.
Let $T(\aD)$ be the transformed diagram of $\aD$. 
Let ${\phi}_u: \reals^2\rightarrow \reals$ be a differentiable probability distribution with mean $u \in \reals^2$ (e.g, the normalized Gaussian where for any $z\in \reals^2$, $\phi_u(z) = \frac{1}{2\pi \tau^2} e^{-\frac{\|z - u\|^2}{2\tau^2}})$. 

\begin{definition}[\cite{Adams2017Persistence}] \label{def:persurface}
Let $\awfunc: \reals^2 \to \reals$ be a non-negative \weightfunc{} for the persistence plane $\reals^2$. Given a persistence diagram $\aD$, its \emph{persistence surface} ${\rho}_\aD: \reals^2 \to \reals$ (w.r.t. $\awfunc$) is defined as: for any $z\in \reals^2$,  
$\rho_\aD(z) = \sum_{u \in T(\aD)} \awfunc(u) \phi_u (z).$ 

The persistence image is a discretization of the persistence surface.
Specifically, fix a grid on a rectangular region in the plane with a collection $\aP$ of $N$ rectangles (pixels). The \emph{persistence image} for a diagram $\aD$ is $\aI_\aD = \{~\aI[\ap]~ \}_{\ap \in \aP}$ consists of $N$ numbers (i.e, a vector in $\mathbb{R}^N$), one for each pixel $\ap$ in the grid $\aP$ with $\aI[\ap] :=\iint_\ap {\rho}_\aD ~dy dx.$ 
\end{definition}

\fullversiononly{The persistence image can be viewed as a vector in $\reals^N$. One can then compute distance between two persistence diagrams $\aD_1$ and $\aD_2$ by the $L_2$-distance $\|\aI_1 - \aI_2\|_2$ between their persistence images (vectors) $\aI_1$ and $\aI_2$. 
The persistence images have several nice properties, including  stability guarantees; see \cite{Adams2017Persistence} for more details. } 

\section{Metric learning frameworks} 
\label{sec:metriclearning}

Suppose we are given a set of $n$ objects $\Xi$ (sampled from a hidden data space $\mathcal{S}$), classified into $k$ classes. 
We want to use these labelled data to learn a good distance for (persistence image representations of) objects from $\Xi$ 
which hopefully is more appropriate at classifying objects in the data space $\mathcal{S}$. 
To do so, below we propose a new persistence-based kernel for persistence images, and then formulate an optimization problem to learn the best \weightfunc{} so as to obtain a good distance metric for $\Xi$ (and data space $\mathcal{S}$). 

\subsection{Weighted persistence image kernel (\PIWK{})}
\label{sec:newkernel}

From now on, we fix the grid $\aP$ (of size $N$) to generate persistence images (so a persistence image is a vector in $\reals^N$). Let $p_s$ be the center of the $s$-th pixel $\ap_s$ in $\aP$, for $s \in \{1,2,\cdots, N\}$. 
We now propose a new kernel for persistence images. A \emph{\weightfunc{}} refers to a non-negative real-valued function on $\reals^2$. 
\begin{definition}\label{def:newkernel}
Let $\ourw: \reals^2 \to \reals$ be a \weightfunc{}. Given two persistence images $\aI$ and $\aI'$, the \emph{($\ourw$-)weighted persistence image kernel (\PIWK)} is defined as: 
     $\kw(\aI, \aI') ~:= \sum_{s=1}^N \ourw(p_s) e^{-\frac{(\aI(s) - \aI'(s))^2}{2\sigma^2}}.$ 
\end{definition}

\fullversiononly{\noindent\underline{\emph{Remark 0:}} We could use the persistence surfaces (instead of persistence images) to define the kernel (with the summation replaced by an integral). Since for computational purpose, one still needs to approximate the integral in the kernel via some discretization, we choose to present our work using persistence images directly. Our Lemma \ref{thm:psd} and Theorem \ref{thm:stability} still hold (with slightly different stability bound) if we use the kernel defined for persistence surfaces. }

\noindent{\underline{\emph{Remark 1:}}} One can choose the \weightfunc{} from different function classes. Two popular choices are: mixture of $m$ 2D Gaussians; and degree-$d$ polynomials on two variables. 

\noindent{\underline{\it Remark 2:}} There are other natural choices for defining a weighted kernel for persistence images. For example, we could use $k(\aI, \aI') = \sum_{s=1}^N e^{-\frac{\ourw(p_s) (\aI(s) - \aI'(s))^2}{2\sigma^2}}$, which we refer this as \emph{\altPIWK}. 
Alternatively, one could use the weight function used in 
PWGK kernel \cite{Kusano2017Kernel} directly. 
Indeed, we have implemented all these choices, and our experiments show that our \PIWK{} kernel leads to better results than these choices for almost all datasets (see Appendix Section 2). 
In addition, note that PWGK kernel \cite{Kusano2017Kernel} contains cross terms $\ourw(x) \cdot \ourw(y)$ in its formulation, meaning that there are quadratic number of terms (w.r.t the number of persistence points) to calculate the kernel, making it more expensive to compute and learn for complex objects (e.g, for the neuron data set, a single neuron tree could produce a persistence diagrams with hundreds of persistence points). 

\begin{lemma}\label{thm:psd}
The \PIWK{} kernel is positive semi-definite. 
\end{lemma}
The rather simple proof of the above lemma is in Appendix Section 1.1. 
By Lemma \ref{thm:psd}, the \PIWK{} kernel gives rise to a Hilbert space. We can now introduce the \PIWK-distance, which is the \emph{pseudo-metric} induced by the inner product on this Hilbert space. 
\begin{definition}\label{def:newdistance}
Given two persistence diagrams $A$ and $B$, let $\aI_A$ and $\aI_B$ be their corresponding persistence images. 
Given a \weightfunc{} $\omega: \reals^2\to \reals$, the \emph{($\omega$-weighted) \PIWK-distance} is:
\begin{align*}
    \ourD(A, B) &:= \sqrt{ \kw(\aI_A, \aI_A) + \kw(\aI_B, \aI_B) - 2 \kw(\aI_A, \aI_B) }. 
\end{align*}
\end{definition}

\paragraph{Stability of \PIWK-distance.}
Given two persistence diagrams $A$ and $B$, two traditional distances between them are the bottleneck distance $d_B(A, B)$ and the $p$-th Wasserstein distance $d_{W,p}(A, B)$. 
\fullversiononly{Stability of these two distances w.r.t. changes of input objects or functions defined on them have been studied \cite{CEH07,CEHY10,CSGO16}.}
Similar to the stability study on persistence images, below we prove \PIWK-distance is stable w.r.t. small perturbation in persistence diagrams as measured by $d_{W,1}$. (Very informally, view two persistence diagrams $A$ and $B$ as two (appropriate) measures (with special care taken to the diagonals), and $d_{W,1}(A, B)$ is roughly the ``earth-mover'' distance between them to convert the measure corresponding to $A$ to that for $B$.)

To simplify the presentation of Theorem \ref{thm:stability}, we use \emph{unweighted persistence images w.r.t. Gaussian}, meaning in Definition \ref{def:persurface}, (1) the weight function $\awfunc$ is the constant function $\awfunc = 1$; and (2) the distribution $\phi_u$ is the Gaussian $\phi_u(z) = \frac{1}{2\pi \tau^2} e^{-\frac{\|z - u\|^2}{2\tau^2}}$. 
(Our result below can be extended to the case where $\phi_u$ is not Gaussian.)
The proof of the theorem below follows from results of \cite{Adams2017Persistence} and can be found in Appendix Section 1.2. 

\begin{theorem}\label{thm:stability}
Given a \weightfunc{} $\ourw: \reals^2 \to \reals$, set $c_w = \| \ourw\|_\infty = \sup_{z\in\reals^2} \ourw(z)$.  
Given two persistence diagrams $A$ and $B$, with corresponding persistence images $\aI_A$ and $\aI_B$, we have that: 
$ \ourD(A, B) \le \sqrt{\frac{20 c_w}{\pi}} \cdot \frac{1}{\sigma \cdot \tau} \cdot d_{W,1}(A, B), 
$
where $\sigma$ is the width of the Gaussian used to define our \PIWK{} kernel (Def. \ref{def:newkernel}), and $\tau$ is that for the Gaussian $\phi_u$ to define persistence images (Def. \ref{def:persurface}). 
\end{theorem}

\fullversiononly{\noindent{\underline{\emph{Remarks:}}} 
We can obtain a more general bound for the case where the distribution $\phi_u$ is not Gaussian. 
Furthermore, we can obtain a similar bound when our \PIWK{}-kernel and its induced \PIWK-distance is defined using \emph{persistence surfaces} instead of \emph{persistence images}.} 

\subsection{Optimization problem for metric-learning}
\label{subsec:optimization}

Suppose we are given a collection of objects $\Xi = \{X_1, \ldots, X_n\}$ (sampled from some hidden data space $\mathcal{S}$), already classified (labeled) to $k$ classes $\aC_1, \ldots, \aC_k$. 
In what follows, we say that $i \in \aC_j$ if $X_i$ has class-label $j$. 
We first compute the persistence diagram $A_i$ for each object $X_i \in \Xi$. (The precise filtration we use to do so will depend on the specific type of objects. Later in Section \ref{sec:exp}, we will describe filtrations used for graph data). 
Let $\{A_1, \ldots, A_n\}$ be the resulting set of persistence diagrams. 
Given a \weightfunc{} $\ourw$, its induced \PIWK{}-distance between $A_i$ and $A_j$ can also be thought of as a distance for the original objects $X_i$ and $X_j$; that is, we can set $\ourD(X_i, X_j) := \ourD(A_i, A_j)$. 
Our goal is to learn a good distance metric for the data space $\mathcal{S}$ (where $\Xi$ are sampled from) from the labels. 
We will formulate this as learning a best \weightfunc{} $\ourw^*$ so that its induced \PIWK{}-distance fits the class-labels of $X_i$'s best. Specifically, for any $t \in \{1,2, \cdots, k\}$, set: 
\begin{align*}
\acost_\ourw(t, t) &= \sum_{i, j \in \aC_t} \ourD^2(A_i, A_j); ~~~~\text{and}~~~~
\acost_\ourw(t, \cdot) ~= \sum_{i\in \aC_t, j \in \{1,2,\cdots, n\}} \ourD^2(A_i, A_j). 
\end{align*}
Intuitively, $\acost_\ourw(t, t)$ is the total in-class (square) distances for $\aC_t$; while $\acost_\ourw(t, \cdot)$ is the total distance from objects in class $\aC_t$ to all objects in $\Xi$. %
A good metric should lead to relatively smaller distance between objects from the same class, but larger distance between objects from different classes. 
We thus propose the following optimization problem, which is related to $k$-way spectral clustering where the distance for an edge $(A_i, A_j)$ is $D_\omega^2(A_i, A_j)$:  
%
\begin{definition}[Optimization problem]\label{def:optproblem}
Given a \weightfunc{} $\ourw: \reals^2 \to \reals$, the \emph{total-cost} of its induced \PIWK{}-distance over $\Xi$ is defined as: 
$\TC(\ourw) := \sum_{t=1}^k \frac{\acost(t, t)}{\acost(t, \cdot)}.$
The \emph{optimal distance problem} aims to find the best \weightfunc{} $\ourw^*$ from a certain function class $\mathcal{F}$ so that the total-cost is minimized; that is: 
$ \TC^* = \min_{\ourw \in \mathcal{F}} \TC(\ourw); ~~\text{and}~~ \ourw^* = \mathrm{argmin}_{\ourw \in \mathcal{F}} \TC(\ourw). 
$
\end{definition}

\paragraph{Matrix view of optimization problem.} 
We observe that our cost function can be re-formulated into a matrix form. This provides us with a perspective from the Laplacian matrix of certain graphs to understand the cost function, and helps to simplify the implementation of our optimization problem, as several programming languages popular in machine learning (e.g Python and Matlab) handle matrix operations more efficiently (than using loops). 
More precisely, recall our input is a set $\Xi$ of $n$ objects with labels from $k$ classes. 
We set up the following matrices: 
\begin{equation*}
\begin{aligned}
L = G - \Lambda; ~~~~\Lambda &= \big[ \Lambda_{ij} \big]_{n \times n}, ~~ \text{where}~ \Lambda_{ij} = \ourD^2(A_i, A_j)~~\text{for}~ i, j\in \{1,2,\cdots,n\}; \\
 G &= \big[ g_{ij} \big]_{n \times n}, \quad \text{where}~~ g_{ij} = 
\begin{cases}
\sum_{\ell = 1}^n \Lambda_{i\ell} & \text{if}~i = j\\
0 & \text{if}~i \neq j
\end{cases}
\\
H &= \big[ h_{ti} \big]_{k \times n} \quad \text{where}~~ h_{ti} = 
\begin{cases}
\frac{1}{\sqrt{\acost_\ourw(t, \cdot) }} & i \in \aC_t\\
0 & otherwise
\end{cases}
\end{aligned}
\end{equation*}

Viewing $\Lambda$ as distance matrix of objects $\{X_1, \ldots, X_n\}$, $L$ is then its Laplacian matrix. 
We have the following main theorem, which essentially is similar to the trace-minimization view of $k$-way spectral clustering (see e.g, Section 6.5 of \cite{kokiopoulou2011trace}).
The proof for our specific setting is in Appendix 1.3. 

\begin{theorem}\label{thm:matrixview}
The total-cost can also be represented by $\TC(\ourw) = k - \trace(HLH^T)$, where $\trace(\cdot)$ is the trace of a matrix. 
Furthermore, $HGH^T = \mathbf{I}$, where $\mathbf{I}$ is the $k\times k$ identity matrix. 
\end{theorem}
Note that all matrices, $L, G, \Lambda,$ and $H$, are dependent on the (parameters of) \weightfunc{} $\ourw$, and in the following corollary of Theorem \ref{thm:matrixview}, we use the subscript of $\ourw$ to emphasize this dependence.
\begin{cor}\label{cor:matrixoptimization}
The Optimal distance problem is equivalent to 
$$\min_{\ourw} \big( k - \trace(H_\ourw L_\ourw H_\ourw^T) \big), ~~~\text{subject to}~~ H_\ourw G_\ourw H_\ourw^T = \mathbf{I}. $$
\end{cor}

\paragraph{Solving the optimization problem.} 
In our implementation, we use (stochastic) gradient descent to find a (locally) optimal \weightfunc{} $\ourw^*$ for the minization problem. 
Specifically, given a collection of objects $\Xi$ with labels from $k$ classes, we first compute their persistence diagrams via appropriate filtrations, and obtain a resulting set of persistence diagrams $\{A_1, \ldots, A_n\}$. 
We then aim to find the best parameters for the \weightfunc{} $\omega^*$ to minimize $Tr(HLH^T) = \sum_{t=1}^k h_t L h_t^T$ subject to $HGH^T = I$ (via Corollary \ref{cor:matrixoptimization}). For example, assume that the \weightfunc{} $\ourw$ is from the class $\mathcal{F}$ of mixture of $m$ number of 2D non-negatively weighted (spherical) Gaussians. 
Each \weightfunc{} $\ourw: \reals^2 \to \reals \in \mathcal{F}$ is thus determined by $4m$ parameters $\{x_r, y_r, \sigma_r, w_r \mid r \in \{1,2,\cdots, m\}\}$ with $\ourw(z) = w_r e^{-\frac{(z_x - x_r)^2 + (z_y - y_r)^2}{\sigma_r^2}}$. 
We then use (stochastic) gradient decent to find the best parameters to minimize $Tr(HLH^T)$ subject to $HGH^T = I$. 
Note that the set of persistence diagrams / images will be fixed through the optimization process. 

From the proof of Theorem \ref{thm:matrixview} (in Appendix 1.3), it turns out that condition $HGH^T = \mathbf{I}$ is satisfied as long as the multiplicative weight $w_r$ of each Gaussian in the mixture is non-negative. Hence during the gradient descent, we only need to make sure that this holds \footnote{
In our implementation, we add a penalty term $\sum_{r=1}^m \frac{c}{exp({w_r})}$ to total-cost $k - Tr(HLH^T)$, to achieve this in a ``soft'' manner.}. 
It is easy to write out the gradient of $\TC(\ourw)$ w.r.t. each parameter $\{x_r, y_r, \sigma_r, w_r \mid r \in \{1,2,\cdots, m\}\}$ {\bf in matrix form}. For example, 
$\frac{\partial \TC(\ourw)}{\partial x_r} = -(\sum_{t=1}^k \frac{\partial h_t}{\partial x_r} L h_t^T + h_t \frac{\partial L}{\partial x_r} h_t^T + h_t L \frac{\partial h_t^T}{\partial x_r});$
where $h_t = \big[ h_{t1}, h_{t2}, ..., h_{tn} \big]$ is the $t$-th row vector of $H$. 
While this does not improve the asymptotic complexity of computing the gradient (compared to using the formulation of cost function in Definition \ref{def:optproblem}), these matrix operations can be implemented much more efficiently than using loops in languages such as Python and Matlab.
For large data sets, we use stochastic gradient decent, by sampling a subset of $s << n$ number of input persistence images, and compute the matrices $H, D, L, G$ as well as the cost using the subsampled data points. The time complexity of one iteration in updating parameters is $O(s^2 N)$, where $N$ is the size of a persistence image (recall, each persistence image is a vector in $\mathbb{R}^N$). 
In our implementation, we use  Armijo-Goldstein line search scheme to update the parameters in each (stochastic) gradient decent step. The optimization procedure terminates when the cost function converges or the number of iterations exceeds a threshold. 
Overall, the time complexity of our optimization procedure is $O(\myitnum s^2 N)$ where $\myitnum$ is the number of iterations, $s$ is the minibatch size, and $N$ is the size ($\#$ pixels) of a single persistence image. 

\section{Experiments}
\label{sec:exp}

We show the effectiveness of our metric-learning framework and the use of the learned metric via graph classification applications. 
In particular, given a set of graphs $\Xi = \{ G_1, \ldots, G_n\}$ coming from $k$ classes, we first compute the unweighted persistence images $A_i$ for each graph $G_i$, and apply the framework from Section \ref{sec:newkernel} to learn the ``best'' \weightfunc{} $\ourw^*: \reals^2 \to \reals$ \fullversiononly{on the birth-death plane $\reals^2$} using these persistence images $\{A_1, \ldots, A_n\}$ and their labels. 
We then perform graph classification using kernel-SVM with the learned $\ourw^*$-\PIWK{} kernel. 
We refer to this framework as \emph{\PIWK{}-classification} framework.  
We show two sets of experiments. 
Section \ref{subsec:neurons} shows that our learned \PIWK{} kernel significantly outperforms existing persistence-based representations. 
In Section \ref{subsec:graphs}, we compare the performance of \PIWK{}-classification framework with various state-of-the-art methods for the graph classification task over a range of data sets. 
More details / results can be found in Appendix Section 2. 

\paragraph{Setup for our \PIWK{}-based framework.}
In all our experiments, we assume that the \weightfunc{} comes from the class $\mathcal{F}$ of mixture of $m$ 2D non-negatively weighted Gaussians as described in the end of Section \ref{subsec:optimization}. 
We take $m$ and the width $\sigma$ in our \PIWK{} kernel as hyperparameters: Specifically, we search among $m\in \{3, 4, 5, 6, 7, 8 \}$ and $\sigma \in \{0.001, 0.01, 0.1, 1, 10, 100 \}$. The $10*10$-fold nested cross validation are applied to evaluate our algorithm: There are 10 folds in outer loop for evaluation of the model with selected hyperparameters and 10 folds in inner loop for hyperparameter tuning. 
We then repeat this process 10 times (although the results are extremely close whether repeating 10 times or not). 
Our optimization procedure terminates when the change of the cost function 
remains $\le 10^{-4}$ or the iteration
number exceeds 2000.  

One important question is to initialize the centers of the Gaussians in our mixture. There are three strategies that we consider. (1) We simply sample $m$ centers in the domain of persistence images randomly. (2) We collect all points in the persistence diagrams $\{A_1, \ldots, A_n\}$ derived from the training data $\Xi$, and perform a k-means algorithm to identify $m$ means. (3) We perform a k-center algorithm to those points to identify $m$ centers. Strategies (2) and (3) usually outperform strategy (1). Thus in what follows we only report results from using k-means and k-centers as initialization, referred to as \emph{\PIWK-kM} and \emph{\PIWK-kC}, respectively. Our code is published in https://github.com/topology474/WKPI.

\begin{table}[tbp]
\centering
{
\scriptsize
\caption{Classification accuracy on neuron dataset. Our results are \PIWK-km and \PIWK-kc. \label{tbl:neurons}}
\begin{tabular}{|c|ccc|cc|cc|}
\hline
Datasets & \multicolumn{3}{|c|}{Existing approaches}& \multicolumn{2}{|c|}{Alternative metric learning} & \multicolumn{2}{c|}{Our \PIWK{} framework}\\
\hline
\quad & PWGK & SW  & PI-PL & \altPIWK & \trainPWGK & \PIWK-km & \PIWK-kc\\
\hline
NEURON-BINARY & 80.5$\pm$0.4 & 85.3$\pm$0.7  & 83.7$\pm$0.3 & 82.1$\pm$2.1 & 84.6$\pm$2.4 & \textcolor{red}{\textbf{89.6 $\pm$2.2}} & 86.4$\pm$2.4\\
NEURON-MULTI & 45.1$\pm$0.3 & 57.6$\pm$0.6 & 44.2$\pm$0.3 & 54.3$\pm$2.3 & 49.7$\pm$2.4 & 56.6$\pm$2.7 & \textcolor{red}{\textbf{59.3$\pm$2.3}}\\
\hline
Average & 62.80 & 71.45 & 63.95 & 68.20 & 67.15 & \textcolor{red}{\textbf{73.10}} & 72.85\\
\hline
\end{tabular}
}
\end{table}
\subsection{Comparison with other persistence-based methods}
\label{subsec:neurons}

We compare our methods with state-of-the-art persistence-based representations, including the Persistence Weighted Gaussian Kernel (PWGK) \cite{Kusano2017Kernel}, original Persistence Image (PI) \cite{Adams2017Persistence}, and Sliced Wasserstein (SW) Kernel \cite{Carri2017Sliced}. Furthermore, as mentioned in \underline{\it Remark 2} after Definition \ref{def:newkernel}, we can learn weight functions in PWGK by the optimizing the same cost function (via replacing our \PIWK-distance with the one computed from PWGK kernel); and we refer to this as \trainPWGK. 
We can also use an alternative kernel for persistence images as described in \underline{\it Remark 2}, and then optimize the same cost function using distance computed from this kernel; we refer to this as \altPIWK{}. 
We will compare our methods both with existing approaches, as well as with these two alternative metric-learning approaches (\trainPWGK{} and \altPIWK{}). 

\paragraph{Generation of persistence diagrams.}
Neuron cells have natural tree morphology, 
rooted at the cell body (soma), 
\fullversiononly{with dendrite and axon branching out,} and  
are commonly modeled as geometric trees. See Figure 1 in the Appendix for an example. 
Given a neuron tree $T$, following \cite{li2017metrics}, we use the descriptor function $f: T \to \reals$ where $f(x)$ is the geodesic distance from $x$ to the root of $T$ along the tree. 
To differentiate the dendrite and axon part of a neuron cell, we further negate the function value if a point $x$ is in the dendrite. We then use the union of persistence diagrams $A_T$ induced by both the sublevel-set and superlevel-set filtrations w.r.t. $f$. 
\fullversiononly{Under these filtrations, intuitively, each point $(b, d)$ in the birth-death plane $\reals^2$ corresponds to the creation and death of certain branch feature for the input neuron tree.}
The set of persistence diagrams obtained this way (one for each neuron tree) is the input to our \PIWK{}-classification framework. 

\paragraph{Results on neuron datasets.}
{\bf Neuron-Binary} dataset consists of 1126 neuron trees from two classes; while {\bf Neuron-Multi} contains 459 neurons from four classes. As the number of trees is not large, we use all training data to compute the gradients in the optimization process instead of mini-batch sampling. 

\fullversiononly{Persistence images are both needed for the methodology of \cite{Adams2017Persistence} and as input for our \PIWK{}-distance, and its resolution is fixed at roughly $40 \times 40$ (see Appendix 2.2 for details).  
For persistence image (PI) approach of \cite{Adams2017Persistence}, we experimented both with the unweighted persistence images (PI-CONST), and one, denoted by (PI-PL), where the weight function $\awfunc: \reals^2 \to \reals$ is a simple piecewise-linear (PL) function adapted from what's proposed in \cite{Adams2017Persistence}; see Appendix 2.2 for details. 
Since PI-PL performs better than PI-CONST on both datasets, Table \ref{tbl:neurons} only shows the results of PI-PL. }
The classification accuracy of various methods is given in Table \ref{tbl:neurons}. 
Our results are consistently better than other topology-based approaches, as well as alternative metric-learning approaches; not only for the neuron datasets as in Table \ref{tbl:neurons}, but also for graph benchmark datasets shown in Table 3 of Appendix Section 2.2, and often by a large margin. 
In Appendix Section 2.1, we also show the heatmaps indicating the learned weight function $\omega: \mathbb{R}^2 \to \mathbb{R}$.

\subsection{Graph classification task}
\label{subsec:graphs} 
\begin{table}[tbp]
\centering
{\scriptsize
\caption{Graph classification accuracy + standard deviation. Our results are last two columns. 
\label{tbl:variance}}
\begin{tabular}{|c|ccccccc|cc|}
\hline
Dataset & \multicolumn{7}{|c|}{Previous approaches} & \multicolumn{2}{|c|}{Our approaches}\\
\hline
 & RetGK & WL  & DGK & P-WL-UC & PF & PSCN & GIN & WKPI-kM & WKPI-kC\\
\hline
{NCI1} & 84.5$\pm$0.2 & 85.4$\pm$0.3 & 80.3$\pm$0.5 & 85.6$\pm$0.3& 81.7$\pm$0.8 & 76.3$\pm$1.7 & 82.7$\pm$1.6 & {\textcolor{red}{\textbf 87.5$\pm$0.5}} & 84.5$\pm$0.5\\
NCI109 &- & 84.5$\pm$0.2 & 80.3$\pm$0.3 & 85.1$\pm$0.3& 78.5$\pm$0.5 & - & - & 85.9$\pm$0.4 & {\textcolor{red}{\textbf 87.4$\pm$0.3}} \\
PTC & 62.5$\pm$1.6 & 55.4$\pm$1.5 & 60.1$\pm$2.5 & 63.5$\pm$1.6& 62.4$\pm$1.8 & 62.3$\pm$5.7 & 66.6$\pm$6.9 & 62.7$\pm$2.7 & {\textcolor{red}{\textbf 68.1$\pm$2.4}} \\
PROTEIN & 75.8$\pm$0.6 & 71.2$\pm$0.8 & 75.7$\pm$0.5 & 75.9$\pm$0.8& 75.2$\pm$2.1 & 75.0$\pm$2.5 & 76.2$\pm$2.6 & {\textcolor{red}{\textbf 78.5$\pm$0.4 }} & 75.2$\pm$0.4 \\
DD & 81.6$\pm$0.3 & 78.6$\pm$0.4 & - & 78.5$\pm$0.4&79.4$\pm$0.8 & 76.2$\pm$2.6 & - & {\textcolor{red}{\textbf 82.0$\pm$0.5 }} & 80.3$\pm$0.4 \\
MUTAG & {\textcolor{red}{\textbf 90.3$\pm$1.1 }} & 84.4$\pm$1.5 & 87.4$\pm$2.7 & 85.2$\pm$0.3 & 85.6$\pm$1.7 & 89.0$\pm$4.4 & 90.0$\pm$8.8 & 85.8$\pm$2.5 & 88.3$\pm$2.6\\
{\tiny IMDB-BINARY} & 71.9$\pm$1.0 & 70.8$\pm$0.5 & 67.0$\pm$0.6 & 73.0$\pm$1.0 & 71.2$\pm$1.0 & 71.0$\pm$2.3 & 75.1$\pm$5.1& 70.7$\pm$1.1 & {\textcolor{red}{\textbf 75.1$\pm$1.1 }} \\
{\tiny IMDB-MULTI} & 47.7$\pm$0.3 & 49.8$\pm$0.5 & 44.6$\pm$0.4 & - & 48.6$\pm$0.7 & 45.2$\pm$2.8 & {\textcolor{red}{\textbf 52.3 $\pm$2.8 }} &46.4$\pm$0.5 & 49.5$\pm$0.4 \\ 
{\tiny REDDIT-5K} & 56.1$\pm$0.5 & 51.2$\pm$0.3 & 41.3$\pm$0.2 & - & 56.2$\pm$1.1 & 49.1$\pm$0.7 & 57.5$\pm$1.5 & 59.1$\pm$0.5 & {\textcolor{red}{\textbf 59.5$\pm$0.6 }}\\
{\tiny {REDDIT-12K}} & {\textcolor{red}{\textbf 48.7$\pm$0.2 }} & 32.6$\pm$0.3 & 32.2$\pm$0.1 & - & 47.6$\pm$0.5 & 41.3$\pm$0.4 & - & 47.4$\pm$0.6 & 48.4$\pm$0.5\\
\hline
\end{tabular}

}
\end{table}

We use a range of benchmark datasets: 
(1) several datasets on graphs derived from small chemical compounds or protein molecules: \textbf{NCI1} and \textbf{NCI109} \cite{shervashidze2011weisfeiler}, \textbf{PTC} \cite{helma2001predictive}, \textbf{PROTEIN} \cite{borgwardt2005protein}, \textbf{DD} \cite{dobson2003distinguishing} and \textbf{MUTAG} \cite{debnath1991structure}; (2) two datasets on graphs representing the response relations between users in Reddit:
\textbf{REDDIT-5K} (5 classes) and \textbf{REDDIT-12K} (11 classes) \cite{Yanardag2015Deep}; 
and (3) two datasets on IMDB networks of actors/actresses: \textbf{IMDB-BINARY} (2 classes), and \textbf{IMDB-MULTI} (3 classes). 
See Appendix Section 2.2 for descriptions of these datasets, and their statistics (sizes of graphs etc). 

Many graph classification methods have been proposed in the literature, with different methods performing better on different datasets. Thus we include multiple approaches to compare with, to include state-of-the-art results on different datasets: 
six graph-kernel based approaches: RetGK\cite{zhang2018retgk}, FGSD\cite{Verma2017}, Weisfeiler-Lehman kernel (WL)\cite{shervashidze2011weisfeiler}, Weisfeiler-Lehman optimal assignment kernel (WL-OA)\cite{Kriege2016On}, Deep Graphlet kernel (DGK)\cite{Yanardag2015Deep}, and the very recent persistent Weisfeiler-Lehman kernel (P-WL-UC) \cite{riech2019persistent}, Persistence Fisher kernel (PF)\cite{Le2018Persistence}; two graph neural networks: PATCHYSAN (PSCN) \cite{Niepert2016Learning}, Graph Isomorphism Network (GIN)\cite{Xu2018b}; and the topology-signature-based neural networks
\cite{Hofer2017Deep}. 

\paragraph{Classification results.}
To generate persistence summaries, we want to put a meaningful descriptor function on input graphs. 
We consider two choices: (a) the {\it Ricci-curvature function} $f_c: G \to \reals$, where $f_c(x)$ is the discrete Ricci curvature for graphs as introduced in \cite{lin2011ricci}; and (b) {\it Jaccard-index function} $f_J: G \to \reals$ which measures edge similarities in a graph. See Appendix 2.2 for details. 
Graph classification results are in Table \ref{tbl:variance}: here Ricci curvature function is used for the small chemical compounds datasets (NCI1, NCI9, PTC and MUTAG), while Jaccard function is used for proteins datasets (PROTEIN and DD) and the social/IMDB networks (IMDB's and REDDIT's).
Results of previous methods are taken from their respective papers -- only a subset of the aforementioned previous methods are included in Table \ref{tbl:variance} due to space limitation. Comparisons with {\bf more methods} are in Appendix Section 2.2. We rerun the two best performing approahes GIN and RetGK using the exactly same nested cross validation setup as ours. The results are also in Appendix Section 2.2, which is similar to those in Table \ref{tbl:variance}. 

Except for \textbf{MUTAG} and {\bf IMDB-MULTI},
the performances of our \PIWK-framework are similar or better than \textbf{the best of other methods}. 
Our \PIWK-framework performs well on both chemical graphs and social graphs, while some of the earlier work tend to work well on one type of the graphs. 
Furthermore, note that the chemical / molecular graphs usually have attributes associated with them. Some existing methods use these attributes in their classification \cite{Yanardag2015Deep, Niepert2016Learning, zhang2018retgk}. Our results however are obtained \textbf{purely based on graph structure} without using any attributes. 
In terms of variance, the standard deviations of our methods tend to be on-par with graph kernel based previous approaches; and are usually much better (smaller) than the GNN based approaches (i.e, PSCN and GIN).

\section{Concluding remarks}
\label{appendix:sec:conclusion}

This paper introduces a new weighted-kernel for persistence images (\PIWK), together with a metric-learning framework to learn the best \weightfunc{} for \PIWK{}-kernel from labelled data. 
Very importantly, we apply the learned \PIWK{}-kernel to the task of graph classification, and show that our new framework achieves similar or better results than the best results among a range of previous graph classification approaches. 

In our current framework, only a single descriptor function of each input object (e.g, a graph) is used to derive a persistence-based representation. It will be interesting to extend our framework to leverage multiple descriptor functions (so as to capture different types of information) simultaneously and effectively. Recent work on multidimensional persistence would be useful in this effort. 
Another interesting question is to study how to incorporate categorical attributes associated to graph nodes (or points in input objects) effectively. 
\fullversiononly{Indeed, real-valued attributed can potentially be used as a descriptor function to generate persistence-based summaries. But the handling of categorical attributes via topological summarization is much more challenging, especially when there is no (prior-known) correlation between these attributes (e.g, the attribute is simply a number from $\{1,2,\cdots, s\}$, coming from $s$ categories. The indices of these categories may carry no meaning).} 

\subsubsection*{Acknowledgments}
The authors would like to thank Chao Chen and Justin Eldridge for useful discussions related to this project. We would also like to thank Giorgio Ascoli for helping provide the neuron dataset.
This work is partially supported by National Science Foundation via grants CCF-1740761, CCF-1733798, and RI-1815697, as well as by National Health Institute under grant R01EB022899. 

{\small
	\bibliographystyle{abbrv}
\bibliography{reference}
}

\newpage

\appendix
\section{Missing details from Section 3}
\label{appendix:sec:newkernel}

\subsection{Proof of Lemma 3.2}
\label{appendix:thm:psd}

Consider an arbitrary collection of $n$ persistence images $\{ \aI_1, \ldots, \aI_n \}$ (i.e, a collection of $n$ vectors in $\reals^N$). 
Set $K = [k_{ij}]_{n\times n}$ to be the $n\times n$ kernel matrix where $k_{ij} = \kw(\aI_i, \aI_j)$. 
Now given any vector $v = (v_{1}, v_{2}, ..., v_{n})^T$, we have that: 
\begin{align*}
v^T K v &= \sum_{i,j = 1}^n v_{i}v_{j} k_{ij} \\
&= \sum_{i,j=1}^n v_{i}v_{j} \sum_{s = 1}^m \ourw(p_s)e^{-\frac{(\aI_i(s) - \aI_j(s))^2}{2 {\sigma}^2}}\\
&= \sum_{s=1}^m \ourw(p_s) \sum_{i,j=1}^n v_{i} v_{j} e^{-\frac{(\aI_i(s) - \aI_j(s))^2}{2 {\sigma}^2}}. 
\end{align*}
Because Gaussian kernel is positive semi-definite and the \weightfunc{} $\ourw$ is non-negative, $v^T K v \geq 0$ for any $v \in \reals^N$. Hence the \PIWK{} kernel is positive semi-definite. 

\subsection{Proof of Theorem 3.4}
\label{appendix:thm:stability}

By Definitions \textbf{3.1} and \textbf{3.3}, combined with the fact that $1 - e^{-x} \le x$ for any $x\in \reals$, we have that: 
\begin{align*}
\ourD^2(A , B) &= \kw(\aI_A, \aI_A) + \kw(\aI_B, \aI_B) - 2\kw(\aI_A, \aI_B)\\
&= ~2 \sum_{s = 1}^N \ourw(p_s) - 2 \sum_{s=1}^n \ourw(p_s) e^{-\frac{(\aI_A(s) - \aI_B(s))^2}{{\sigma}^2}}\\
&= ~2 \sum_{s = 1}^N \ourw(p_s)(1 - e^{-\frac{(\aI_A(s) - \aI_B(s))^2}{{\sigma}^2}}) \\
& ~\leq 2 c_w \sum_{s = 1}^N (1 - e^{-\frac{(\aI_A(s) - \aI_B(s))^2}{{\sigma}^2}})\\
& ~\leq 2 \frac{c_w}{{\sigma} ^ 2} \sum_{s = 1}^n (\aI_A(s) - \aI_B(s))^2 \\
& ~\leq 2 \frac{c_w}{{\sigma} ^ 2} ||\aI_A - \aI_B||_2^2
\end{align*}

Furthermore, by Theorem 10 of \cite{Adams2017Persistence}, when the distribution $\phi_u$ to in Definition 2.1 is the normalized Gaussian $\phi_u(z) = \frac{1}{2\pi \tau^2} e^{-\frac{\|z - u\|^2}{2\tau^2}}$, and the weight function $\awfunc = 1$, we have that 
$\| \aI_A - \aI_B \|_2 \le \sqrt{\frac{10}{\pi}} \cdot \frac{1}{\tau} \cdot d_{W, 1} (A, B)$. (Intuitively, view two persistence diagrams $A$ and $B$ as two (appropriate) measures, and $d_{W,1}(A, B)$ is then the ``earth-mover'' distance between them so as to convert the measure corresponding to $A$ to that for $B$, where the cost is measured by the total $L_1$-distance that all mass have to travel.)
Combining this with the inequalities for $\ourD^2(A,B)$ above, the theorem then follows. 

\subsection{Proof of Theorem \textbf{3.6}}
\label{appendix:thm:matrixview}

We first show the following properties of matrix $L$ which will be useful for the proof later. 

\begin{lemma}\label{lem:Lproperties}
The matrix L is symmetric and positive semi-definite. 
Furthermore, for every vector $f \in \reals^n$, we have 
\begin{equation}\label{eqn:Lpsd}
f^T L f = \frac{1}{2} \sum_{i,j = 1}^n \Lambda_{ij} (f_i - f_j)^2
\end{equation}
\end{lemma}
\begin{proof} 
By construction, it is easy to see that $L$ is symmetric as matrices $\Lambda$ and $G$ are. 
The positive semi-definiteness follows from Eqn (\ref{eqn:Lpsd}) which we prove now. 
\begin{equation*}
\begin{aligned}
f^T L f &= f^T G f - f^T \Lambda f ~~= \sum_{i=1}^n f_i^2 g_{ii} - \sum_{i, j = 1}^n f_i f_j \Lambda_{ij} \\
&= \frac{1}{2}\big(\sum_{i=1}^n f_i^2 g_{ii} +\sum_{j=1}^n f_j^2 g_{jj} - \sum_{i, j = 1}^n 2 f_i f_j \Lambda_{ij} \big) \\
&= \frac{1}{2}\big(\sum_{i=1}^n f_i^2 \sum_{j=1}^n \Lambda_{ij} + \sum_{j=1}^n f_j^2 \sum_{i=1}^n \Lambda_{ji}\\
&~~~~~~~~~~~~- \sum_{i, j = 1}^n 2 f_i f_j \Lambda_{ij} \big) \\
&= \frac{1}{2} \sum_{i,j = 1}^n \Lambda_{ij} \cdot (f_i^2 + f_j^2 - 2 f_i f_j) \\
&= \frac{1}{2} \sum_{i,j=1}^n \Lambda_{ij}(f_i - f_j)^2
\end{aligned}
\end{equation*}
The lemma then follows. 
\end{proof}

We now prove the statement in Theorem 3.6. 
Recall that the definition of various matrices, and that $h_t$'s are the row vectors of matrix $H$. For simplicity, in the derivations below, we use $D(i,j)$ to denote the $\ourw$-induced \PIWK{}-distance $\ourD(A_i, A_j)$ between persistence diagrams $A_i$ and $A_j$. 
Applying Lemma \ref{lem:Lproperties}, we have: 
\begin{equation}\label{eqn:traceone}
\begin{aligned}
\trace&(HLH^T) = \sum_{t = 1}^k (HLH^T)_{tt} 
~~= \sum_{t=1}^k h_t L h_t^T \\
&= \sum_{t=1}^k \frac{1}{2} \cdot \sum_{j_1, j_2 = 1}^n D^2(j_1, j_2)(h_{t, j_1} - h_{t, j_2})^2 \\
& = \sum_{t=1}^k \frac{1}{2} \cdot \sum_{j_1, j_2 = 1}^n D^2(j_1, j_2) (h_{t, j_1}^2 + h_{t, j_2}^2 - 2 h_{t, j_1} h_{t, j_2}). 
\end{aligned}
\end{equation}
Now by definition of $h_{ti}$, it is non-zero only when $i\in \aC_t$. Combined with Eqn (\ref{eqn:traceone}), it then follows that: 
\begin{equation*}
    \begin{aligned}
\trace&(HLH^T) = \sum_{t=1}^k \frac{1}{2} \cdot \big(\sum_{j_1\in \aC_t, j_2 \in [1,n]} \frac{D^2(j_1, j_2)}{\acost_\ourw(t, \cdot) } \\
&+ \sum_{j_1\in [1,n], j_2 \in \aC_t} \frac{D^2(j_1, j_2)}{\acost_\ourw(t, \cdot) } - 2 \sum_{j_1, j_2 \in \aC_t} \frac{D^2(j_1, j_2)}{\acost_\ourw(t, \cdot) } \big) \\
&= \sum_{t=1}^k \frac{1}{2} \big(\sum_{j_1 \in \aC_t, j_2 \notin \aC_t}\frac{D^2(j_1, j_2)}{\acost_\ourw(t, \cdot) } \\
&+ \sum_{j_1 \notin \aC_t, j_2 \in \aC_t} \frac{D^2(j_1, j_2)}{\acost_\ourw(t, \cdot) }\big)\\
&= \sum_{t=1}^k \sum_{j_1 \in A_t, j_2 \notin A_t} \frac{D^2(j_1, j_2)}{\acost_\ourw(t, \cdot) } \\
& = \sum_{t=1}^k \frac{\acost_\ourw(t, \cdot)  - \acost_\ourw(t, t) }{\acost_\ourw(t, \cdot) } \\
&= k - \TC(\ourw)
\end{aligned}
\end{equation*}
This proves the first statement in Theorem 3.6. 
We now show that the matrix $HGH^T$ is the $k\times k$ identity matrix $\mathbf{I}$. 
Specifically, first consider $s \neq t \in [1, k]$; we claim:  
\begin{equation*}
\begin{aligned}
(H G H^T)_{st} &= h_s G h_t^T = \sum_{j_1, j_2 = 1}^n h_{sj_1} G_{j_1 j_2} h_{tj_2} \quad =0. 
\end{aligned}
\end{equation*}
It equals to $0$ because $h_{sj_1}$ is non-zero only for $j_1 \in \aC_s$, while $h_{tj_2}$ is non-zero only for $j_2 \in \aC_t$. However, for such a pair of $j_1$ and $j_2$, obviously $j_1 \neq j_2$, which means that $G_{j_1j_2} = 0$. Hence the sum is $0$ for all possible $j_1$ and $j_2$'s. 

Now for the diagonal entries of the matrix $HGH^T$, we have that for any $t\in [1, k]$:  
\begin{equation*}
\begin{aligned}
(H G H^T)_{tt} &= h_t G h_t^T = \sum_{j_1, j_2 = 1}^n h_{t j_1} G_{j_1, j_2} h_{tj_2} \\
& = \sum_{j_1, j_2 \in \aC_t} \frac{G_{j_1 j_2}}{\acost_\ourw(t, \cdot) } \quad = \sum_{j_1 \in \aC_t} \frac{G_{j_1 j_1}}{\acost_\ourw(t, \cdot) }\\
& = \sum_{j_1 \in \aC_t} \frac{\sum_{\ell =1}^n D^2(j_1, \ell)}{\acost_\ourw(t, \cdot) } \\
& = \frac{\sum_{j_1 \in \aC_t, \ell \in [1, n]} D^2(j_1, \ell)}{\acost_\ourw(t, \cdot) } \\
&= \frac{\acost_\ourw(t, \cdot) }{\acost_\ourw(t, \cdot) } \quad = 1. 
\end{aligned}
\end{equation*}

This finishes the proof that $HGH^T = \mathbf{I}$, and completes the proof of Theorem 3.6.

\section{More details for Experiments}
\label{appendix:sec:exp}

\subsection{More on neuron experiments}
\label{appendix:subsec:neurons}

\paragraph{Description of neuron datasets. }
Neuron cells have natural tree morphology (see Figure \ref{fig:neuron} (a) for an example), 
rooted at the cell body (soma), with dentrite and axon branching out. Furthermore, this tree morphology is important in understanding neurons. Hence it is common in the field of neuronscience to model a neuron as a (geometric) tree (see Figure \ref{fig:neuron} (b) for an example downloaded from NeuroMorpho.Org\cite{ascoli2007neuromorpho}). 

Our NEURON-BINARY dataset consists of 1126 neuron trees classified into two (primary) classes: \emph{interneuron} and \emph{principal neurons} (data partly from the Blue Brain Project \cite{markram2015reconstruction} and downloaded from http://neuromorpho.org/). The second NEURON-MULTI dataset is a refinement of the 459 interneuron class into four (secondary) classes: \emph{basket-large, basket-nest, neuglia} and \emph{martino}. 
\begin{figure}[tbhp]
\begin{center}
    \begin{tabular}{cc}
    \includegraphics[height = 2.5cm]{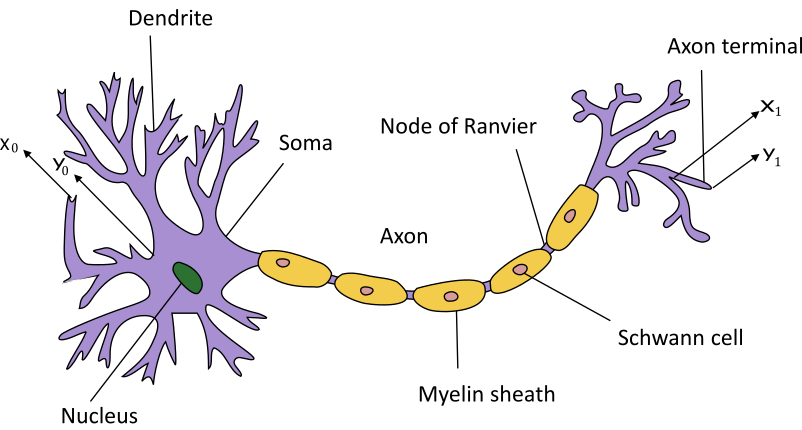} & 
    \includegraphics[height = 3cm, width = 3cm]{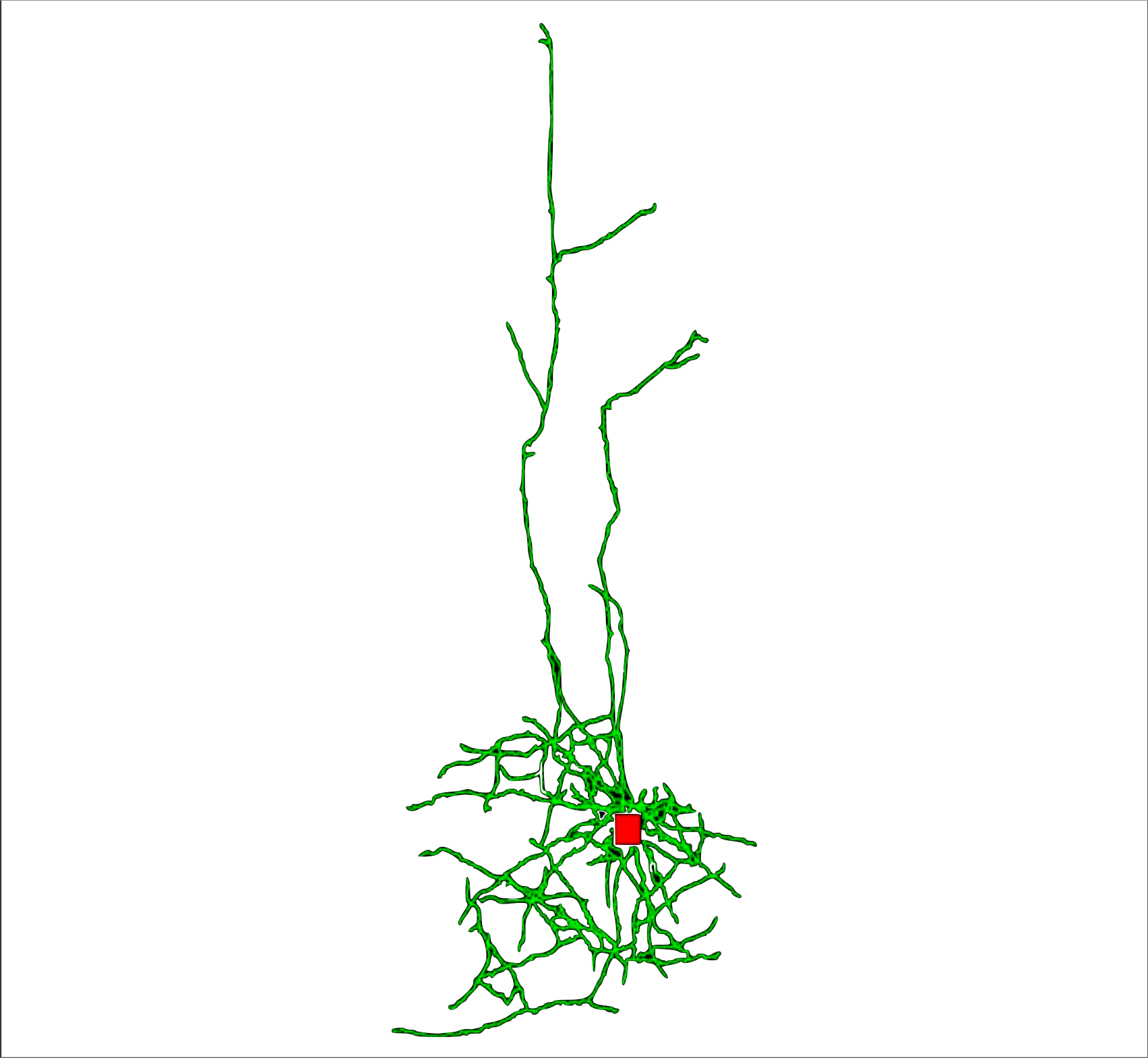}\\
    (a) &  (b) 
    \end{tabular}
\end{center}
\vspace*{-0.15in}
\caption{(a) An neuron cell (downloaded from Wikipedia)and (b) an example of a neuron tree (downloaded from NeuroMorpho.Org). \label{fig:neuron}}
\end{figure}

\paragraph{Setup for persistence images.}
Persistence-images are both needed for the methodology of \cite{Adams2017Persistence} and as input for our \PIWK{}-distance. 
%
For each dataset, the persistence image for each object inside is computed within the rectangular bounding box of the points from all persistence diagrams of input trees. The $y$-direction is then discretized to $40$ uniform intervals, while the $x$-direction is discretized accordingly so that each pixel is a square. 
For persistence image (PI) approach of \cite{Adams2017Persistence}, we show results both for the unweighted persistence images (PI-CONST), and one, denoted by PI-PL, where the weight function $\awfunc: \reals^2 \to \reals$ (for Definition 2.1) is the following piecewise-linear function (modified from one proposed by Adams et al. \cite{Adams2017Persistence}) where $b$ the largest persistence for any persistent-point among all persistence diagrams. 
\begin{equation}
\awfunc(x,y) = 
\begin{cases}
\frac{|y-x|}{b} & |y-x| < b ~ and ~ y > 0\\
\frac{|-y-x|}{b} & |-y-x| < b ~ and ~ y < 0\\
1 & otherwise
\end{cases}
\end{equation}

\paragraph{Weight function learnt.}
In Figure \ref{fig:neuronheatmap} we show the heatmaps of the learned \weightfunc{} $\ourw^*$ for both datasets. Interestingly, we note that the important branching features (points in the birth-death plane with high $\ourw^*$ values) separating the two primary classes (i.e, for {\bf Neuron-Binary} dataset) is different from those important for classifying neurons from one of the two primary classes (the interneuron class) into the four secondary classes (i.e, the {\bf Neuron-Multi} dataset). Also high importance (weight) points may not have high persistence. 
In the future, it would be interesting to investigate whether the important branch features are also biochemically important. 

\begin{figure}[tbhp]
\begin{center}
    \begin{tabular}{cc}
    \includegraphics[height = 3cm]{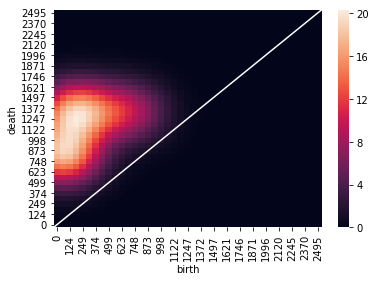}
    &  \includegraphics[height = 3cm]{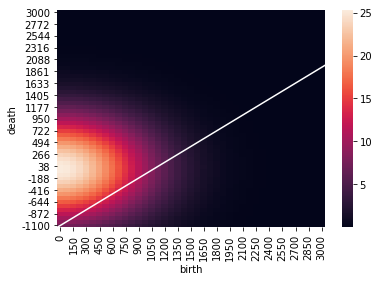}\\
    \end{tabular}    
\end{center}
\vspace*{-0.15in}
\caption{ Heatmaps of the learned \weightfunc{} $\ourw^*$ for Neuron-Binary (left) and Neuron-Multi (right) datasets. Each point in this plane indicates the birth-death of some branching feature. Warmer color (e.g, red) indicates higher $\ourw^*$ value. $x$- and $y$-axies are birth / death time measured by the descriptor function $f$ (modified geodesic function, where for points in dendrites they are negation of the distance to root).} 
\label{fig:neuronheatmap}
\end{figure}

\subsection{More on graph classification experiments}

\paragraph{Benchmark datasets for graph classification.} 
Below we first give a brief description of the benchmark datasets we used in our experiments. These are collected from the literature. 

\textbf{NCI1} and \textbf{NCI109} \cite{shervashidze2011weisfeiler} consist of two balanced subsets of datasets of chemical compounds screened for activity against non-small cell lung cancer and ovarian cancer cell lines, respectively.\\
\textbf{PTC} \cite{helma2001predictive} is a dataset of graph structures of chemical molecules from rats and mice which is designed for the predictive toxicology challenge 2000-2001.\\
\textbf{DD} \cite{dobson2003distinguishing} is a data set of 1178 protein structures. Each protein is represented by a graph, in which the nodes are amino acids and two nodes are connected by an edge if they are less
than 6 Angstroms apart. They are classified according to whether they are enzymes or not.\\
\textbf{PROTEINS} \cite{borgwardt2005protein} contains graphs of protein. In each graph, a node represents a secondary structure element (SSE) within protein structure, i.e. helices, sheets and turns. Edges connect nodes if they are neighbours along amino acid sequence or neighbours in protein structure space. Every node is connected to its three nearest spatial neighbours.\\
\textbf{MUTAG} \cite{debnath1991structure} is a dataset collecting 188 mutagenic aromatic and heteroaromatic nitro compounds labelled according to whether they have a mutagenic effect on the Gramnegtive bacterium Salmonella typhimurium.\\
\textbf{REDDIT-5K} and \textbf{REDDIT-12K} \cite{Yanardag2015Deep} consist of graph representing the discussions on the online forum Reddit. In these datasets, nodes represent users and edges between two nodes represent whether one of these two users leave comments to the other or not. In REDDIT-5K, graphs are collected from 5 sub-forums, and they are labelled by to which sub-forums they belong. In REDDIT-12K, there are 11 sub-forums involved, and the labels are similar to those in REDDIT-5K.\\
\textbf{IMDB-BINARY} and \textbf{IMDB-MULTI} \cite{Yanardag2015Deep} are dataset consists of networks of 1000 actors or actresses who played roles in movies in IMDB. In each graph, a node represents an actor or actress, and an edge connects two nodes when they appear in the same movie. In IMDB-BINARY, graphs are classified into Action and Romance genres. In IMDB-MULTI, they are collected from three different genres: Comedy, Romance and Sci-Fi. 

The statistics of these datasets are provided in Table \ref{tbl:datastats}. 
In our experiments, for REDDIT-12K dataset, due to the larger size of the dataset (with about 13K graphs), we deploy the EigenPro method (\cite{ma2017diving}, code available at https://github.com/EigenPro/EigenPro-matlab), which is a preconditioned (stochastic) gradient descent iteration) to significantly improve the efficiency of kernel-SVM. 

\begin{table*}[htbp]
\centering
\caption{Statistics of the benchmark graph datasets}
\vspace{0.1in}
\begin{tabular}{cccccccccc}
\hline
Dataset & \#classes & \#graphs & average \#nodes & average \#edges\\
\hline
NCI1 & 2 & 4110 & 29.87 & 32.30\\
NCI109 & 2 & 4127 & 29.68 & 31.96\\
PTC & 2 & 344 & 14.29 & 14.69\\
PROTEIN & 2 & 1113 & 39.06 & 72.82\\
DD & 2 & 1178 & 284.32 & 715.66\\
IMDB-BINARY & 2 & 1000 & 19.77 & 96.53\\
IMDB-MULTI & 3 & 1500 & 13.00 & 65.94\\
REDDIT-5K & 5 & 4999 & 508.82 & 594.87\\
REDDIT-12K & 11 & 12929 & 391.41 & 456.89\\
\hline
\label{tbl:datastats}
\end{tabular}
\end{table*}

\paragraph{Persistence generation.}
To generate persistence diagram summaries, we want to put a meaningful descriptor function on input graphs. 
We consider two choices in our experiments: (a) the {\it Ricci-curvature function} $f_c: G \to \reals$, where $f_c(x)$ is a discrete Ricci curvature for graphs as introduced in \cite{lin2011ricci}; and (b) {\it Jaccard-index function} $f_J: G \to \reals$. 

Then Ollivier's Ricci curvature between two nodes $u$ and $v$ is $\kappa^{\alpha}_{uv} = 1 - W(m_u^{\alpha}, m_v^{\alpha}) / d(u,v)$ where $W(\cdot, \cdot)$ is Wasserstein distance between two measures and $d(u,v)$ is the distance between two nodes, and probability measure $m_u^{\alpha}$ around node $u$ is defined as 
 \begin{equation}
     m_x^{\alpha}(x) =
          \begin{cases}
    \alpha & x = u\\
     (1 - \alpha)/n_u & x \in \mathcal{N}(u)\\
     0 & \text{otherwise}
     \end{cases}
 \end{equation}
 $n_u = |\mathcal{N}(u)|$ and $\alpha$ is a parameter within $[0, 1]$. In this paper, we set $\alpha = 0.5$.

In particular, the Jaccard-index of an edge $(u, v)\in G$ in the graph is defined as $\rho(u,v) = \frac{|NN(u) \cap NN(v)|}{|NN(u) \cup NN(v)|}$, where $NN(x)$ refers to the set of neighbors of node $x$ in $G$. The Jaccard index has been commonly used as a way to measure edge-similarity\footnote{We modify our persistence algorithm slightly to handle the edge-valued Jaccard index function}. 
As in the case for neuron data sets, we take the union of the $0$-th persistence diagrams induced by both the sublevel-set and the superlevel-set filtrations of the descriptor function $f$, and convert it to a persistence image as input to our \PIWK{}-classification framework \footnote{We expect that using the $0$-th zigzag persistence diagrams will provide better results. However, we choose to use only $0$-th standard persistence as it can be easily implemented to run in $O(n\log n)$ time using a simple union-find data structure.}. 

In all results reported in main text and in Table \ref{tbl:graphs}, Ricci curvature function is used for the small chemical compounds data sets (NCI1, NCI9, PTC and MUTAG), while Jaccard function is used for the two proteins datasets (PROTEIN and DD) as well as the social/IMDB networks (IMDB's and REDDIT's). Both 0-dim and 1-dim extented persistence diagrams are employed. In general, we observe that Ricci curvature is more sensitive to accurate graph local structure, while Jaccard function is better for noisy graphs (with noisy edge). In Figure \ref{fig: benchmarkheatmap}, we show the heatmaps of the weight function before and after our metric learning for NCI1 and REDDIT-5K datasets. In particular, the left column shows the heatmaps of the initialized weight function, while the right column shows the heatmaps of the optimal weight function as learned by our algorithm.

\begin{table*}[htbp]
\centering
{\scriptsize
\caption{Classification accuracy on graphs. Our results are in columns \PIWK{}-kM and \PIWK-kC. 
\label{tbl:graphs}}
\vspace{0.1in}
\begin{tabular}{|c|cccccccc|cc|}
\hline
Dataset & \multicolumn{8}{|c|}{Previous approaches} & \multicolumn{2}{|c|}{Our appraches}\\
\hline
  & RetGK & WL & WL-OA & DGK & FGSD & PSCN & GIN & P-WL-UC & \PIWK-kM & \PIWK-kC\\
\hline
NCI1 & 84.5 & 85.4 & 86.1 & 80.3 & 79.8 & 76.3 & 82.7 & 85.6 & \textcolor{red}{\textbf{87.5}} & 84.5\\
NCI109 &- & 84.5 & 86.3 & 80.3 & 78.8 & - & - & 85.1 & 85.9 & \textcolor{red}{\textbf{87.4}}\\
PTC & 62.5 & 55.4 & 63.6 & 60.1 & 62.8 & 62.3 & 66.6 & 63.5 & 62.7 & \textcolor{red}{\textbf{68.1}}\\
PROTEIN & 75.8 & 71.2 & 76.4 & 75.7 & 72.4 & 75.0 & 76.2 & 75.9 &\textcolor{red}{\textbf{78.5}} & 75.2\\
DD & 81.6 & 78.6 & 79.2 & - & 77.1 & 76.2 & - & 78.5 & \textcolor{red}{\textbf{82.0}}& 80.3\\
MUTAG & 90.3 & 84.4 & 84.5 & 87.4 & \textcolor{red}{\textbf{92.1}} & 89 & 90 & 85.2 & 85.8 & 88.3\\
IMDB-BINARY & 71.9 & 70.8 & - & 67.0 & 71.0 & 71.0 & 75.1 & 73.0 & 70.7 & \textcolor{red}{\textbf{75.4}}\\
IMDB-MULTI & 47.7 & 49.8 & -  & 44.6 & 45.2 & 45.2 & \textcolor{red}{\textbf{52.3}}& - & 46.4 & 49.5\\ 
REDDIT-5K & 56.1 & 51.2 & - & 41.3 & 47.8 & 49.1 & 57.5 & - & 59.1 & \textcolor{red}{\textbf{59.5}}\\
REDDIT-12K & \textcolor{red}{48.7} & 32.6 & - & 32.2 & - & 41.3 & - & - & 47.4 & 48.4 \\
\hline
Average & - & 66.39 &- & -&- &- &- & - & 69.99 & {\textcolor{red}{\textbf 71.66}}  \\
\hline
\end{tabular}
}

\end{table*}

\begin{figure}[tbhp]
    \centering
    \begin{tabular}{cc}
      ~~Initial weight function & ~~Learnt weight function \\
    \includegraphics[height = 3cm]{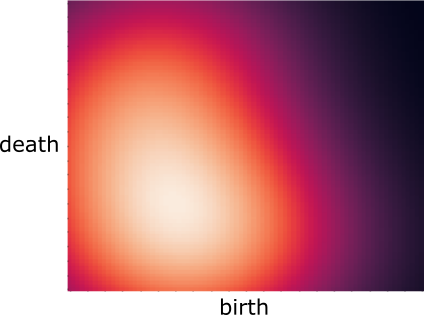} & \includegraphics[height = 3cm]{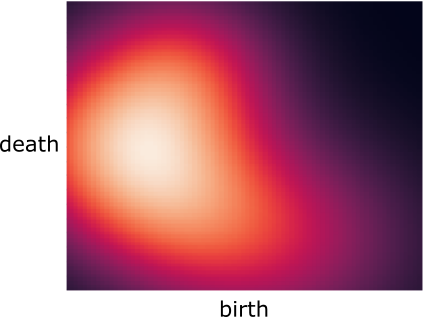} \\
       \includegraphics[height = 3cm]{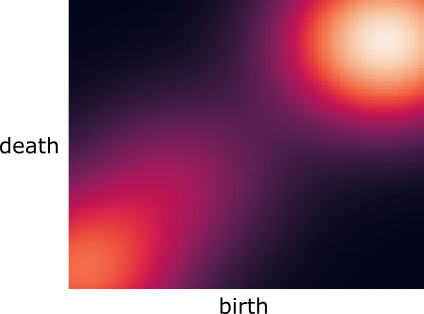} & 
        \includegraphics[height = 3cm]{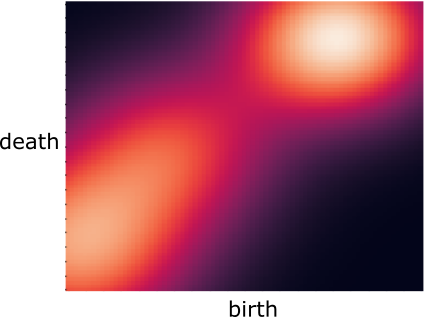}
    \end{tabular}
    \caption{Heatmap of initialized weight function (left column) and that of the learnt \weightfunc{} $\ourw^*$  (right column). Top row shows results for NCI1 data set; while bottom row contains those for REDDIT-5K data set.}
    \label{fig: benchmarkheatmap}
\end{figure}

\paragraph{Additional results.} 
Many graph classification methods have been proposed in the literature. We compare our results with a range of existing approaches, which includes state-of-the-art results on different datasets:  
six graph-kernel based approaches: RetGK\cite{zhang2018retgk}, FGSD\cite{Verma2017}, Weisfeiler-Lehman kernel (WL)\cite{shervashidze2011weisfeiler}, Weisfeiler-Lehman optimal assignment kernel (WL-OA)\cite{Kriege2016On}, Deep Graphlet kernel (DGK)\cite{Yanardag2015Deep}, and the very recent persistent Weisfeiler-Lehman kernel (P-WL-UC)\footnote{Note that results for three version of persistent WL kernels are reported in their paper. We take the one (P-WL-UC, with uniform node labels) that performs the best from their Table 1.} \cite{riech2019persistent}; two graph neural networks: PATCHYSAN (PSCN) \cite{Niepert2016Learning}, Graph Isomorphism Network (GIN)\cite{Xu2018b}; as well as the topology-signature-based neural networks
\cite{Hofer2017Deep}. 

Additional results of comparing our results with more existing methods are given in Table \ref{tbl:graphs}. 
The results of DL-TDA (topological signature based deep learning framework) \cite{Hofer2017Deep} are not listed in Table \ref{tbl:graphs}, as only the classification accuracy for REDDIT-5K (accuracy $54.5\%$) and REDDIT-12K ($44.5\%$) are given in their paper (although their paper contains many more results on other objects, such as images). While also not listed in this table, we note that our results also outperform the newly independently proposed general neural network architecture for persistence representations reported in the very recent preprint \cite{carriere2019general}.  
Comparison with other topological-based non-neural network approaches are given below.

\paragraph{Topological-based methods on graph data.} 
Here we compare our \PIWK{}-framework with the performance of several state-of-the-art persistence-based classification frameworks, including: PWGK \cite{Kusano2017Kernel}, SW \cite{Carri2017Sliced}, PI \cite{Adams2017Persistence} and PF \cite{Le2018Persistence}. We also compare it with two alternative ways to learn the metric for persistence-based representations: 
{\bf \trainPWGK{}} is the version of PWGK \cite{Kusano2017Kernel} where we learn the weight function in its formulation, using the same cost-function as what we propose in this paper for our \PIWK{} kernel functions. 
{\bf \altPIWK{}} is the alternative formulation of a kernel for persistence images where we set the kernel to be $k(\aI, \aI') = \sum_{s=1}^N e^{-\frac{\ourw(p_s) (\aI(s) - \aI'(s))^2}{2\sigma^2}}$, instead of our \PIWK{}-kernel as defined in Definition 3.1. 
\begin{table*}[htbp]
\centering
{\scriptsize 
\caption{Classification accuracy on graphs for topology-based methods. 
\label{tbl:persbenchmark}}
\vspace{0.1in}
\begin{tabular}{|c|ccccc|cc|cc|}
\hline
Datasets & \multicolumn{5}{|c|}{Existing TDA approaches} & \multicolumn{2}{|c|}{Alternative metric learning} & \multicolumn{2}{c|}{Our \PIWK{} framework}\\
\hline
 & PWGK & PI-CONST & PI-PL & SW & PF & trainPWGK & altWKPI & WKPI-kM & WKPI-kC\\
\hline
NCI1 & 73.3 & 72.5 & 72.1 & 80.1 & 81.7 & 76.5 & 77.4 & \textcolor{red}{\textbf{87.2}} & 84.7\\
NCI109 & 71.5 & 74.3 & 73.1 & 75.5 & 78.5 & 77.2 & 81.2 & 85.5 & \textcolor{red}{\textbf{86.9}}\\
PTC & 62.2 & 61.3 & 64.2 & \textcolor{red}{\textbf{64.5}} & 62.4 & 62.5 & 64.2 & 61.1 & 64.3\\
PROTEIN & 73.6 & 72.2 & 69.1 & 76.4 & 75.2 & 74.8 & 75.1 & \textcolor{red}{\textbf{77.4}} & 75.6\\
DD & 75.2 & 74.2 & 76.8 & 78.9 & 79.4 & 76.4 & 72.5 & \textcolor{red}{\textbf{79.8}} & 79.1 \\
MUTAG & 82.0 & 85.2 & 83.5 & 87.1 & 85.6 & 86.4 & \textcolor{red}{\textbf{88.5}} & 85.5 & 88.0\\
IMDB-BINARY & 66.8 & 65.5 & 69.7 & 69.6 & 71.2 & 71.8 & 67.3 & 70.6 & \textcolor{red}{\textbf{75.4}}\\
IMDB-MULTI & 43.4 & 42.5 & 46.4 & 48.7 & 48.6 & 45.8 & 45.3 & 47.1 & \textcolor{red}{\textbf{48.8}} \\
REDDIT-5K & 47.6 & 52.2 & 51.7 & 53.8 & 56.2 & 53.5 & 54.7 & 58.7 & \textcolor{red}{\textbf{59.3}} \\
REDDIT-12K & 38.5 & 43.3 & 45.7 & \textcolor{red}{\textbf{48.3}} & 47.6 & 43.7 & 42.1 & 45.2 & 44.5\\
\hline
Average & 63.41 & 64.3 & 65.23 & 68.29 & 68.64 & 66.86 & 66.83 & 69.81 & \textcolor{red}{\textbf{70.66}}\\
\hline
\end{tabular}
}
\end{table*}

\begin{table*}[htbp]
    \centering
    
{\scriptsize 
\caption{Graph classification accuracy of GIN and RetGK on graph benchmarks with the same nested cross validation setup}
\begin{tabular}{ccccccccccc}
\hline
 & NCI1 & NCI109 & PTC & PROTEIN & DD & MUTAG & IMDB-BIN & IMDB-MULTI & Reddit5K & Reddit12K\\
\hline
RetGK & 84.5$\pm$0.2 & 84.8$\pm$0.2 & 62.9$\pm$1.6 & 75.4$\pm$0.6 & 81.6$\pm$0.4 & 90.0$\pm$1.1 & 72.3$\pm$1.0 & 47.7$\pm$0.4 & 55.8 $\pm$0.5 & 48.5$\pm$ 0.2\\
GIN & 82.4$\pm$1.6 & 86.5$\pm$1.5 & 67.8$\pm$6.5 & 76.7$\pm$2.6 & 81.1$\pm$2.5 & 89.0$\pm$7.5 & 75.6$\pm$5.3 & 52.4$\pm$3.1  & 57.2$\pm$1.5 &           47.9$\pm$ 2.1 \\
\hline
\end{tabular}
}
    
    \label{tbl:GINRETGK}
\end{table*}

We use the same setup as our WKPI-framework to train these two metrics, and use their resulting kernels for SVM to classify the benchmark graph datasets. WKPI-framework outperforms the existing approaches and alternative metric learning methods on all datasets except \textbf{MUTAG}. WKPI-kM (i.e, WKPI-kmeans) and KWPI-kC (i.e, WKPI-kcenter) improve the accuracy by $3.9\% - 11.9\%$ and $5.4 \% - 13.5\%$, respectively. Besides, we show results by another experimental setup. In 10-fold cross validation, choose $m$ and $\sigma$ leading to the smallest cost function value, then evaluate the classifier on the test set. Repeat this process 10 times. That is, $m$ and $\sigma$ are not the hyperparameters of the SVM classifiers, but are determined by the metrics learning. We refer to these two approaches as WKPI-kM and WKPI-kC in accordance with the initialization methods. The classification accuracy of all these methods are reported in Table \ref{tbl:persbenchmark}. 

\end{document}